\documentclass{amsart}

\newtheorem{theorem}{Theorem}[section]
\newtheorem{lemma}[theorem]{Lemma}

\theoremstyle{definition}

\newtheorem{xca}[theorem]{Exercise}

\theoremstyle{remark}

\numberwithin{equation}{section}

\begin{document}

\title{Convex and linear models of NP-problems}

\author{Sergey Gubin}
\email{sgubin@genesyslab.com}






\begin{abstract}
Reducing the NP-problems to the convex/linear analysis on the Birkhoff polytope.
\end{abstract}

\maketitle

\section*{Introduction}
\label{Introduction}

Since the classical works of J. Edmonds \cite[and others]{edmonds}, linear modeling became a common technique in combinatorial optimization \cite[and others]{kp, yan, bond, tsp, schrijver, graph}. Often, the linear models are expressed with some constrains on the incidence vector. The major benefit of this approach is the symmetry of the resulting model: the resulting equations are an invariant under relabeling. The major disadvantage of the approach is difficulty to express the constrains explicitly due to their size and structural complexity \cite[and others]{kp, yan, bond}.
\newline\indent
In this work, the Subgraph Isomorphism Problem \cite{cook0, karp, sf, gi} is taken as a basic NP-problem. The adjacency and incidence matrices are used to express the linear and convex models explicitly. In such asymmetric models, the unknown is a relabeling. The relabeling is presented with an unknown permutation matrix. That reduces the NP-problems to the linear/convex analysis on the Birkhoff polytope \cite{bir}. 
 
\section{Adjacency matrix models}
\label{Inequalities}

Let's take the Subgraph Isomorphism Problem \cite{cook0, karp, sf}: whether a given multi digraph $g$ contains a subgraph which is isomorphic with another given multi digraph $s$. That is a NP-complete problem.
\newline\indent
Let $n$ and $m$ be powers of vertex-sets of $g$ and $s$, appropriately. Based on a node labeling/enumeration, let's construct adjacency matrices of these digraphs  - matrices $G$ and $S$, appropriately. In terms of these matrices, the problem is a compatibility problem for the following quadratic system where the unknown is permutation matrix $X = (x_{ij})_{n\times n}$:
\begin{equation}
\label{eIn1}
P_{mn} X^T G X P_{mn}^{T} \geq S \\
\end{equation}
Here, matrix $P_{mn}$ is a truncation:
$$
P_{mn} = (U_{m} ~ 0)_{m\times n},
$$
- where matrix $U_{m}$ is the union matrix of size $m\times m$.
\newline\indent 
Permutation matrix $X$ presents all $n!$ possible ways to label vertices of (multi) digraph $g$. Compatibility of system \ref{eIn1} means that there is such a vertex labeling of digraph $g$ that the appropriate adjacency matrix of that digraph will ``cover'' the adjacency matrix of (multi) digraph $s$. Such a labeling of $g$ would make the positive solution for instance $(g,s)$ the self evident. Obviously, if the system is incompatible, then instance $(g,s)$ has solution ``NO''. 
\newline\indent
Model \ref{eIn1} is asymmetric \cite{yan} because relabeling of $g$ rotates solutions of system \ref{eIn1} over all permutation matrices. The matrices are vertices of the Birkhoff polytope.
\newline\indent
Let's call matrices $G$ and $S$ the problem's instance and pattern, appropriately. The matrix pairs (or more precisely, the pair's conjugacy classes over the group of permutation matrices) parametrize the whole NP zoo. Let us illustrate that with several examples \cite{edmonds, cook0, karp, sf, gi, tsp, schrijver, graph}:
\begin{description}
\item[(Sub)GI] 
Instance and pattern are adjacency matrices of given (multi di-) graphs. In case of the GI problem: $P_{mn} = U_n$ and sign ``$\geq$'' is replaced with sign ``=''.
\item[Clique] 
Instance is an adjacency matrix of a given (multi di-) graph. Pattern is a matrix whose diagonal elements are $0$ and remaining elements are $1$:
\[
S = (1)_{m\times m} - U_m = \left (\begin{array}{cccc}
0& 1&\ldots&1 \\
1&0& \ldots&1 \\
\vdots&&\ddots &\vdots\\
1&1&\ldots  &0 \\
\end{array} \right )_{m\times m}
\]
\item[HC]
Instance is adjacency matrix of a given (multi di-) graph. Pattern is any circular permutation matrix, for example: 
\begin{equation}
\label{e:hc}
S = \left (\begin{array}{ccccc}
0& 0&\ldots &0&1 \\
1&0& \ldots&0 &0 \\
0&1&\ldots&0 &0 \\
\vdots&&\ddots &\ddots &\vdots\\
0&0&\ldots       &1&0 \\
\end{array} \right )_{n\times n}
\end{equation}
\item[HP]
The same as the above except one 1 is poked out:
\[
S = \left (\begin{array}{ccccc}
0& 0&\ldots &0&0 \\
1&0& \ldots&0 &0 \\
0&1&\ldots&0 &0 \\
\vdots&&\ddots &\ddots &\vdots\\
0&0&\ldots       &1&0 \\
\end{array} \right )_{n\times n}
\]
\item[Matching]
In this case, $m$ is an even number. Pattern $S$ can be, for example, the following matrix:
\[
S = \left( \begin{array}{cccccc}
0&1&0&0&0&\ldots \\
0&0&0&0&0&\ldots \\
0&0&0&1&0&\ldots \\
0&0&0&0&0&\ldots \\
\vdots&\vdots&\vdots&\vdots&\ddots&\ddots\\
\end{array} \right )_{m \times m},
\]
- where $1$ and $0$ are altering on the over-diagonal.
\item[Perfect Matching]
The same as the above except $m = n$.
\item[3-SAT]
Let $f$ be a 3-SAT instance:
\[
f = \bigwedge_{i=1}^{n} L_{i1} \vee L_{i2} \vee L_{i3}.
\]
Let's arbitrarily enumerate strings in the truth tables of the given clauses. By definition, two strings are compatible if they are consistent and equal $true$:
\[
T(i,\alpha,j,\beta) = (L_{i1} \vee L_{i2} \vee L_{i3})|_{\tau_{i\alpha}}~\wedge ~ (L_{j1} \vee L_{j2} \vee L_{j3})|_{\tau_{j\beta}} = true,
\]
- where $\tau_{xy}$ is $y$-th truth assignment for $x$-th clause, i.e. $y$-th strings in the truth table for $x$-th clause. Compatibility box for clauses $i$ and $j$ is the following matrix $8\times 8$:
\[
B_{ij} = (T(i,\alpha,j,\beta))_{8\times 8}.
\]
The strings' compatibility matrix is the instance:
\[
G = (B_{ij})_{n\times n},
\]
- where diagonal boxes are the union matrices of size $8\times 8$. The pattern is a box matrix with $8\times 8$ boxes: all elements in the boxes are $0$ except $(1,1)$-elements which are $1$:
\[
C_{ij} = \left( \begin{array}{ccc}
1&0&\ldots \\
0&0&\ldots \\
\vdots&\vdots&\ddots\\
\end{array} \right )_{8 \times 8}
\]
Compatibility of system \ref{eIn1} means that there is a true assignment satisfying the given 3-SAT instance.
\item[2-SAT]
The same as the above, except the boxes are $4 \times 4$.
\item[SAT]
Let $f$ be a SAT instance:
\[
f = \bigwedge_{i=1}^{n} c_i, ~ c_i = \bigvee_{\alpha=1}^{n_i} L_{i\alpha}
\]
Let's enumerate literals in each clause. Let's build a compatibility box for every two clauses. The box is a rectangular matrix whose elements are $0$ or $1$ depending on whether the appropriate literals in the clauses are complimentary:
\[
B_{ij} = (1 - \delta(L_{i\alpha}, \bar{L}_{j\beta}))_{n_i \times n_j},
\]
- where $\delta(a,b)$ is the Kronecker delta:
\[
\delta(a,b) = \left \{ \begin{array}{cl}
1, & a = b \\
0, & a \neq b \\
\end{array} \right.
\]
Compatibility matrix is a box matrix of the compatibility boxes. Each clause is presented in this box matrix with one box-row and one box-column with the same indexes. The compatibility matrix is the instance:
\[
G = (B_{ij})_{n\times n}.
\]
- where diagonal boxes $B_{ii}$ are the union matrices of size $n_i\times n_i$. Pattern is a box matrix with the same structure but the boxes are filled with $0$ except their upper left elements which are $1$:
\[
C_{ij} = \left( \begin{array}{ccc}
1&0&\ldots \\
0&0&\ldots \\
\vdots&\vdots&\ddots\\
\end{array} \right )_{n_i \times n_j}
\]
Compatibility of system \ref{eIn1} means that the disjunctive normal form of the SAT-instance has an implicant.
\end{description}
The examples are examples of reduction to the Subgraph Isomorphism Problem.
\newline\indent
System \ref{eIn1} allows different exact and approximate methods to reduce the number of options (the number of non-deterministic solutions to check) or to solve the problem.
\newline\indent
For example, let pair $(G,S)$ define a Clique instance. Because matrix multiplication is a combination of operations ``+'' and ``$\times$'' and all matrices involved in system \ref{eIn1} are the non-negative matrices, inequalities \ref{eIn1} will be true for powers of matrices $G$ and $S$, as well. Thus, graph $g$ can be depleted by comparing elements of matrix $G^k$ with the value of non-diagonal elements of matrix $S^k$, $k \geq 2$. For digraphs and $k=2$, the procedure can be sketched as follows:
\begin{description}
\item[Step 1]
Compute matrix $G^2 = (a_{ij})_{n\times n}$;
\item[Step 2]
Rid of all those edges $(i,j)$ for which 
\[
a_{ij} < m-2 ~\vee ~ a_{ji} < m-2.
\]
\end{description}
Iteration of this procedure $O(n-m)$ times can reduce the instance's dimension or even solve the instance. The method can be modified for the Maximum Clique. One would start with pattern matrix $S$ of size $n\times n$ and then reduce the pattern's size using the procedure.

\subsection{Convex models}

Let's fulfill pattern graph $s$ with $n-m$ isolated vertices. Then, sizes of matrices $G$ and $S$ will be equal and the truncation matrix in formula \ref{eIn1} will become the union matrix:
$$
n = m ~ \Rightarrow ~ P_{mn} = U_n.
$$
Because all matrices involved in system \ref{eIn1} are the non-negative matrices and matrix multiplication involves only summation and multiplication of numbers, the system can be rewritten as follows:
$$
G X \geq X S \\
$$
Let's relax the last system by replacing permutation matrix $X$ with a double stochastic matrix:
\begin{equation}
\label{eIn121}
\left \{ \begin{array}{l}
G X \geq X S \\
\sum_{i} x_{ij} \leq 1 \\
\sum_{j} x_{ij} \leq 1 \\
x_{ij} \geq 0 \\
\end{array} \right.
\end{equation}
\begin{theorem}
\label{t:convex}
NP-instance $(G,S)$ has solution ``YES'' iff value $\sqrt{n}$ is the solution of the following convex program:
\begin{equation}
\label{eIn122}
\sqrt{\sum_{ij} x^2_{ij}} ~ \rightarrow ~ \max,
\end{equation}
- under constrains \ref{eIn121}.
\end{theorem}
\begin{proof}
Due to the Birkhoff-von Neumann theorem about the double stochastic matrices \cite{bir}, value $\sqrt{n}$ is the solution of program \ref{eIn122} on the Birkhoff polytope
\begin{equation}
\label{eIn123}
\sum_{i} x_{ij} = 1, ~ \sum_{j} x_{ij} = 1, ~ x_{ij} \geq 0.
\end{equation}
Matrix $X$ delivers the maximum iff matrix $X$ is a permutation matrix. Permutation matrices are the extreme points of the Birkhoff polytope \ref{eIn123}.
\end{proof}
Convex program \ref{eIn122} under constrains \ref{eIn121} can be solved in polynomial time by the ellipsoid method \cite{nem1} or by the inner point method \cite{nem2}.
\newline\indent
Convex program \ref{eIn122} under constrains \ref{eIn121} is an asymmetric polynomial size model for NP-problems. The asymmetry is due to the fixed vertex labeling in which adjacency matrix $G$ is written. Relabeling of digraph $g$ will rotate the program's solutions over all vertices of the Birkhoff polytope \ref{eIn123}.
\newline\indent
If a NP-problem instance has the negative solution - solution ``NO'', - then, due to theorem \ref{t:convex}, there are two cases: the problem's matrix $G$ does not satisfy system \ref{eIn121}, at all (an easy case); or, due to the Birkhoff-von Neumann theorem about double stochastic matrices and the Carathodory theorem about convex hull, the only solutions of the system are certain convex combinations of $\alpha$ permutation matrices:
\begin{equation}
\label{eIn124}
2 \leq \alpha \leq (n-1)^2 + 1,
\end{equation}
- because the Birkhoff polytope has dimension $(n-1)^2$.
\newline\indent
Obviously, constrains \ref{eIn121} in theorem \ref{t:convex} may be replaced with the following system:
\begin{equation}
\label{eIn125}
\left \{ \begin{array}{l}
X^T G \geq S X^T \\
\sum_{i} x_{ij} \leq 1 \\
\sum_{j} x_{ij} \leq 1 \\
x_{ij} \geq 0 \\
\end{array} \right.
\end{equation}
\indent
Theorem \ref{t:convex} reduces the Subgraph Isomorphism Problem to the convex programming which is a P-problem \cite{nem1,nem2}. The theorem can be re-formulated:
\begin{theorem}
\label{t:convex2}
NP-instance $(G,S)$ has solution ``YES'' iff the following convex system is compatible:
\[
\left \{ \begin{array}{l}
\sqrt{\sum_{ij} x^2_{ij}} = \sqrt{n} \\
\\
G X \geq X S \\
\sum_{i} x_{ij} \leq 1 \\
\sum_{j} x_{ij} \leq 1 \\
x_{ij} \geq 0 \\
\end{array} \right.
\]
\end{theorem}
\begin{proof}
Value $\sqrt{n}$ is the Euclidean norm of the permutation matrices of size $n\times n$.
\end{proof}

\subsection{Asymmetric linear models}
\label{s:asymmetric}

The following example shows that compatibility of systems \ref{eIn121} and \ref{eIn125} is not sufficient for making decisions.
\begin{xca}
Let the instance and pattern matrices be as follows:
$$
G = \left (\begin{array}{cc}
0 & 1 \\
1 & 0 \\
\end{array} \right ),
~
S = \left (\begin{array}{cc}
1 & 0 \\
0 & 1 \\
\end{array} \right )
$$
The only permutation matrices for $n=2$ are
$$
X_1 = \left (\begin{array}{cc}
0 & 1 \\
1 & 0 \\
\end{array} \right ),
~
X_2 = \left (\begin{array}{cc}
1 & 0 \\
0 & 1 \\
\end{array} \right )
$$
Half-sum of the matrices satisfies systems \ref{eIn121} and \ref{eIn125}:
$$
\frac{X_1}{2} + \frac{X_2}{2} = (\frac{X_1}{2} + \frac{X_2}{2})^T = \frac{1}{2} \left (\begin{array}{cc}
1 & 1 \\
1 & 1 \\
\end{array} \right )
$$
$$
S(\frac{X_1}{2} + \frac{X_2}{2}) = (\frac{X_1}{2} + \frac{X_2}{2})S = \frac{1}{2}\left (\begin{array}{cc}
1 & 1 \\
1 & 1 \\
\end{array} \right )
$$
$$
G(\frac{X_1}{2} + \frac{X_2}{2}) = (\frac{X_1}{2} + \frac{X_2}{2})G = \frac{1}{2} \left (\begin{array}{cc}
1 & 1 \\
1 & 1 \\
\end{array} \right )
$$
Nevertheless, instance $(G,S)$ has solution ``NO'': value of criterion \ref{eIn122} on the solution is $1 < \sqrt{2}$.
\end{xca}
The following lemma clarifies the example.
\begin{lemma}
Let pattern $S$ be a permutation matrix. Let $\sigma$ be the set of all permutation matrices $Y$ which have the following property:
$$
Y \leq G.
$$
NP-problem instance $(G,S)$ has solution ``YES'' iff 
$$
\exists~Y \in \sigma: Y \in Cl(S).
$$
\end{lemma}
\begin{proof}
$$
Cl(\sigma) = \bigcup_{Y \in \sigma} Cl(Y).
$$
\end{proof} 
For example 1,
$$
G = X_2 \notin Cl(S) = \{X_1\}.
$$
\indent
Let $L$ be any linear functional on $R^{n^2}$. Then, the following asymmetric linear program models NP-problems:
\begin{equation}
\label{e:lp}
L(X) ~ \rightarrow ~ \max,
\end{equation}
- under constrains \ref{eIn121} or \ref{eIn125}. NP-instance $(G,S)$ has solution ``YES'' iff there are the constrains' extreme points among the optimums of the program.
\newline\indent
If system \ref{eIn121} or system \ref{eIn125} has no solutions at all, then instance $(G,S)$ has solution ``NO''. If the systems have a solution $X$ and that solution is not a permutation matrix, then the solution is a double stochastic matrix. Then, in accordance with the Birkhoff - von Neumann theorem:
\[
X \in \mbox{conv}\{X_1,X_2,\ldots,X_\alpha\},
\]
- where $X_i$ are permutation matrices and $\alpha$ is in bounds \ref{eIn124}. The presentation is not unique, but it takes $O(n^2)$ time to find a permutation matrix which participates in one of the presentations:
\begin{description}
\item[Step 1]
Select any non-zero element in matrix $X$. Let $x_{i_1j_1}$ be the selection;
\item[Step 2]
Select any non-zero element which is not in row $i_1$ nor in column $j_1$. Such element exists because $X$ is a double stochastic matrix. Let $x_{i_2j_2}$ be the selection;
\item[Step 3]
Select any non-zero element which is not in rows $i_1,~ i_2$ nor in columns $j_1, ~j_2$. Such element exists because $X$ is a double stochastic matrix. Let $x_{i_3j_3}$ be the selection;
\item[Steps 4$\div n$]
And so on until $n$ elements will be selected. 
\item[Step $n+1$]
Replace the selected elements with $1$ and replace the rest of elements with $0$. The resulting matrix $X_1$ is a permutation matrix which participates in the presentation of $X$ as a convex combination of permutation matrices.
\end{description}
If permutation matrix $X_1$ is a solution of system \ref{eIn121} or system \ref{eIn125}, then NP-instance $(G,S)$ has solution ``YES''. Otherwise, the systems can be fulfilled with the following inequality:
\[
\sum_{\mu=1}^n x_{i_\mu j_\mu} \leq n-2,
\]
- where $i_\mu$ and $j_\mu$ are the indexes selected in the above procedure. In accordance with the ellipsoid/separation method \cite{gls}, the inequality can be used for the next iteration of the ellipsoid method \cite{Khach}. 

\subsection{A symmetric linear model}

Let's fulfill pattern graph $s$ with $n-m$ isolated vertices. Then, sizes of matrices $G$ and $S$ will be equal and the truncation matrix in formula \ref{eIn1} will become the union matrix. Then, system \ref{eIn1} can be rewritten:
\begin{equation}
\label{eIn111}
G \geq X S X^T
\end{equation}
- where permutation matrix $X$ is the unknown.
\newline\indent
Let's enumerate all permutation matrices:
$$
X_1, ~ X_2, ~ \ldots ~ X_{n!}.
$$
Then, integer quadratic system \ref{eIn111} can be replaced with the following integer linear system:
\begin{equation}
\label{eIn112}
\left \{ \begin{array}{l}
\sum_{i}\lambda_{i}X_{i}S X_{i}^T \leq G \\
\\
\lambda_i \in \{0,1\} \\
\end{array} \right.
\end{equation}
- where numbers $\lambda_i$ are the unknown. The system can be relaxed.
\begin{theorem}
\label{t:symmetric}
System \ref{eIn112} is compatible iff the following system is compatible:
\begin{equation}
\label{eIn113}
\left \{ \begin{array}{l}
\sum_{i}\lambda_{i}X_{i}S X_{i}^T \leq G \\
\\
\sum_i \lambda_i = 1, ~ \lambda_i \geq 0 \\
\end{array} \right.
\end{equation}
- where numbers $\lambda_i$ are the unknown.
\end{theorem}
\begin{proof}
Let numbers $\lambda_i$ solve system \ref{eIn112} and $\lambda_{i_1} = 1$. Then, the following numbers are a solution of system \ref{eIn113}:
$$
\lambda_i = \left \{ \begin{array}{cl}
1, & i = i_1 \\
0, & i \neq i_1 \\
\end{array} \right.
$$ 
Let numbers $\lambda_i$ solve system \ref{eIn113} and $\lambda_{i_1} \neq 0$. Because all matrices participated in \ref{eIn113} are the $(0,1)$-matrices,
$$
G \geq X_{i_1} S X_{i_1}^T.
$$
Thus, the following numbers are a solution of system \ref{eIn112}:
$$
\lambda_i = \left \{ \begin{array}{cl}
1, & i = i_1 \\
0, & i \neq i_1 \\
\end{array} \right.
$$ 
\end{proof}
Due to theorem \ref{t:symmetric}, the following linear program will solve NP-problem \ref{eIn1}:
\begin{equation}
\label{eIn114}
\sum_i \lambda_i ~ \rightarrow ~ \min_{\lambda_1,\lambda_2,\ldots,\lambda_{n!}},
\end{equation}
- under constrains \ref{eIn113}. The program (or its dual, more precisely) can be tried and solved with the ellipsoid/separation method \cite{Khach,yan,gls}. The separation conditions can be arranged with inequalities
\[
X_i SX_i^T \leq G.
\]
Due to estimation \ref{eIn124}, the addends on the left side of system \ref{eIn113} may be analyzed in chunks of size $(n-1)^2 + 1$.
\newline\indent
Linear program \ref{eIn114}/\ref{eIn113} is a symmetric $n!$-size linear program. The symmetry \cite{yan} is due to the explicit involvement of all permutation matrices in constrains \ref{eIn113}.
\newline\indent
For the HC problem, pattern matrix $S$ is a circular permutation matrix. In this case, constrains \ref{eIn113} are an explicit expression for the TSP polytope. The ATSP with a weight matrix $W$ can be modeled as the following linear program:
$$
(W,~\sum_{i}\lambda_{i}X_i S X_i^T) \rightarrow \min_{\lambda_i},
$$ 
- under constrains \ref{eIn113}. The matrix scalar product $(*,*)$ totals products of the appropriate elements of its multiplicands.
\newline\indent
According to the Yannakakis theorem \cite{yan}, size of the TSP polytope for the symmetric linear program has to be bigger than polynomial. System \ref{eIn113} shows that $n!$ is an upper bound for the size.

\subsection{Miscellaneous}
\label{mis1}

\begin{lemma} 
NP-instance $(G,S)$ has solution ``NO'' iff the following system is incompatible for any permutation matrix $R$:
$$
\left \{ \begin{array}{l}
G X \geq X S \\
X \geq R \\
\sum_{i} x_{ij} \leq 1 \\
\sum_{j} x_{ij} \leq 1 \\
x_{ij} \geq 0 \\
\end{array} \right.
$$
\end{lemma}
\begin{proof}
If instance $(G,S)$ has solution ``YES'', there is permutation matrix $X$ satisfying the system for permutation matrix $R=X$. Let matrix $X$ satisfy the system for some permutation matrix $R$. $X$ is a double stochastic matrix. Due to the Birkhoff-von Neumann theorem, $X$ is a convex combination of several permutation matrices. Then, inequality $X \geq R$ implies that $X$ is a permutation matrix: $X=R$. Then, problem $(G,S)$ has solution ``YES''. 
\end{proof}
\begin{lemma}
Let $(G,S)$ be such an instance that matrices $G$ and $S$ allow the following decomposition:
$$
\begin{array}{ll}
G \geq G_1G_2, & G_1 \geq (0)_{n\times n}, ~ G_2 \geq (0)_{n\times n} \\
S \geq S_1S_2, & S_1 \geq (0)_{n\times n}, ~ S_2 \geq (0)_{n\times n} \\
\end{array}
$$
The instance has solution ``YES'' if the following linear system is compatible:
$$
\left \{ \begin{array}{l}
G_1 \geq X S_1 \\
G_2 \geq S_2 X^T \\
\sum_{i} x_{ij} \leq 1 \\
\sum_{j} x_{ij} \leq 1 \\
x_{ij} \geq 0 \\
\end{array} \right.
$$
\end{lemma}
\begin{proof}
Let matrix $X$ satisfy the system. Then, $X$ is a convex combination of several permutation matrices $X_i$:
$$
X \in \mbox{conv } \{X_i~|~1 \leq i \leq \alpha \}.
$$
Then,
$$
G \geq G_1G_2 \geq X S_1S_2 X^T \geq X S X^T \geq X_1 S X_1^T,
$$
\end{proof}
Any NP-problem (matrix $S$) has $2^{n^2}$ different instances (matrices $G$). Ultimately, if there would exist the mathematical tables of the ``YES''/''NO''-instances, the tables could be sorted in a way. The binary sorting would reduce the NP-problems to the binary search. The computational complexity of the search would be
$$
O(\log_2 2^{n^2}) = O(n^2).
$$
The ``oracle-tables'' might be a digital/analog computer which would solve, for example, convex program \ref{eIn122}/\ref{eIn121}.

\section{Incidence matrices models}
\label{Equations}

Let's arbitrarily label/enumerate elements of adjacency matrices $G$ and $S$: if an element is equal $a \geq 0$, then that element has $a$ labels (zero-elements have no labels). Let's construct in-incidence matrix $I_{G} = (\alpha_{ij})$: $\alpha_{ij} = 1$ if $i$-th column of $G$ contains $j$-th label; and $\alpha_{ij} = 0$, if otherwise. In the same way, let's construct in-incidence matrix $I_{S}$ for matrix $S$. Also, let's construct out-incidence matrices $O_{G}$ and $O_{S}$ but using rows instead of columns. Direct calculation proves the following decompositions: 
\begin{equation}
\label{eEq0}
G = O_{G}I_{G}^{T}, ~ S = O_{S}I_{S}^{T}.
\end{equation}
\indent
Let $k$ be the total of all elements of $G$; and $l$ be the total of all elements of $S$.  Then matrices $I_{G}$ and $O_{G}$ are $n\times k$; and matrices $I_{S}$ and $O_{S}$ are $m\times l$. In digraph terms: numbers $k$ and $l$ are powers of the arc-sets of (multi) digraphs $g$ and $s$, appropriately \footnote{The decomposition can be done for any rectangular matrix.}.
\newline\indent
The incidence matrices are total unimodular matrices. They are $(0,1)$-matrices with the following structure: there is one and only one $1$ per column. The $0$-rows in out-incidence matrix indicate sinks in the digraph, and the $0$-rows in in-incidence matrix indicate sources. The isolated vertices are presented with $0$-rows in both incidence matrices. From this point of view, the isolated vertices are sinks and sources simultaneously.
\newline\indent
In terms of the incidence matrices, system \ref{eIn1} can be rewritten as follows:
\begin{equation}
\label{eEq1}
\left \{ \begin{array}{l}
P_{mn}XO_{G}ZP_{l k}^{T} = O_{S} \\
\\
P_{mn}XI_{G}ZP_{lk}^{T} = I_{S} \\
\end{array}, \right.
\end{equation}
- where permutation matrices $X=(x_{ij})_{n\times n}$ and $Z=(z_{ij})_{k\times k}$ are the unknown. In digraph terms: $X$ presents all $n!$ ways to label vertices of $g$, and $Z$ presents all $k!$ ways to label arcs of $g$.
\newline\indent
Due to the unimodularity of the incidence matrices, each of the two parts of system \ref{eEq1} has an integral solution if it has a solution at all. The point is the existence of such a common integral solution that
\[
XX^T = U_n, ~ ZZ^T = U_k.
\]
\indent
Let's arbitrarily enumerate all permutation matrices of size $n\times n$ and $k\times k$:
$$
X_1, X_2, \ldots, X_{n!}, ~ Z_1, Z_2, \ldots, Z_{k!}
$$
Let's write the following system:
\begin{equation}
\label{eEq2}
\left \{ \begin{array}{l}
\sum_{i,j} \lambda_{ij}P_{mn}X_{i}O_{G}Z_{j}P_{lk}^{T} = O_{S} \\
\\
\sum_{i,j} \lambda_{ij}P_{mn}X_{i}I_{G}Z_{j}P_{lk}^{T} = I_{S} \\
\\
\lambda_{ij} \in \{0,1\} \\
\end{array} \right.
\end{equation}
- where numbers $\lambda_{ij}$ are unknown.
\begin{lemma}
System \ref{eEq1} has a permutation matrices solution iff system \ref{eEq2} has such a solution $\lambda_{ij}$ that
\begin{equation}
\label{eEq3}
\exists ~ i_1, j_1:~ 
\lambda_{ij} = \left \{ \begin{array}{cl}
1, & i = i_1 ~ \wedge ~ j = j_1 \\
0, & i \neq i_1 ~ \vee ~ j \neq j_1 \\
\end{array} \right.
\end{equation}
\end{lemma}
\begin{proof}
Let permutation matrices $X$ and $Z$ be a solution of system \ref{eEq1}. Then, the following numbers are a solution of system \ref{eEq2}:
$$
\lambda_{ij} = \left \{ \begin{array}{cl}
1, & X_i = X ~ \wedge ~ Z_j = Z \\
0, & X_i \neq X ~ \vee ~ Z_j \neq Z \\
\end{array} \right.
$$
And visa-versa, if a solution of system \ref{eEq2} has structure \ref{eEq3}, then permutation matrices $X_{i_1}$ and $Z_{j_1}$ are a solution of system \ref{eEq1}.
\end{proof}
One might think that solutions \ref{eEq3} are the only solutions of system \ref{eEq2} possible. That is incorrect.
\begin{xca}
Let instance $(G,S)$ be as follows:
$$
O_G = \left ( \begin{array}{cc}
1 & 0 \\
0 & 1 \\
\end{array} \right ),
~ 
I_G = \left ( \begin{array}{cc}
0 & 1 \\
1 & 0 \\
\end{array} \right )
$$
$$
O_S = (1), ~ I_S = (1)
$$
The only permutation matrices for $n=2$ are 
$$
X_1 = \left ( \begin{array}{cc}
0 & 1 \\
1 & 0 \\
\end{array} \right ),
~
X_2 = \left ( \begin{array}{cc}
1 & 0 \\
0 & 1 \\
\end{array} \right )
$$
The following numbers are a solution of system \ref{eEq2} for the instance:
$$
\lambda_{11} = 1, ~ \lambda_{12} = 1, ~ \lambda_{21} = 0, ~ \lambda_{22} = 0.
$$
Really,
$$
P_{12}X_1 O_G X_1P_{12}^T  = (1), ~ P_{12}X_1 O_G X_2P_{12}^T  = (0)
$$
$$
P_{12}X_1 I_G X_1P_{12}^T  = (0), ~ P_{12}X_1 O_G X_2P_{12}^T  = (1)
$$
\end{xca}
Let's add $n-m$ isolated vertices to pattern graph $s$. That will make $P_{mn} = U_n$ in systems \ref{eEq1} and \ref{eEq2}. Then, system \ref{eEq1} can be rewritten as follows:
\begin{equation}
\label{eEq10}
\left \{ \begin{array}{l}
XO_{G}ZP_{l k}^{T} = O_{S} \\
\\
XI_{G}ZP_{lk}^{T} = I_{S} \\
\end{array}, \right.
\end{equation}
- and system \ref{eEq2} becomes as follows:
\begin{equation}
\label{eEq4}
\left \{ \begin{array}{l}
\sum_{i,j} \lambda_{ij}X_{i}O_{G}Z_{j}P_{lk}^{T} = O_{S} \\
\\
\sum_{i,j} \lambda_{ij}X_{i}I_{G}Z_{j}P_{lk}^{T} = I_{S} \\
\\
\lambda_{ij} \in \{0,1\} \\
\end{array} \right.
\end{equation}
\begin{lemma}
\label{l:structure}
Any solution of system \ref{eEq4} has structure \ref{eEq3}.
\end{lemma}
\begin{proof}
All matrices $O_G$, $O_S$, $I_G$, and $I_S$ have one and only one $1$ per column. But any solution of system \ref{eEq4} with more than one $\lambda_{ij} = 1$ will produce, on the left side of the system, either a non-$(0,1)$-matrix or a $(0,1)$-matrix with more than one $1$ per column.
\end{proof}
System \ref{eEq4} is a symmetric integer linear model of the NP-problems. The  symmetry is due to the explicit involvement of all $n!k!$ combinations of $n\times n$ and $k\times k$ permutation matrices. 
\newline\indent
For a given NP-instance, iteration of the following procedure can significantly reduce the instance's dimension and even solve the instance in polynomial time:
\begin{description}
\item[Step 1]
Build linear combinations of the equations of system \ref{eEq4} in order to make the right sides of the combinations equal to $0$;
\item[Step 2]
Rid system \ref{eEq4} of all those $\lambda_{ij}$ which are on the left sides of the combinations because, due to lemma \ref{l:structure}, they all are $0$ (the ridding can be partial when it is difficult to track all $\lambda_{ij} = 0$).
\end{description}
From this point of view, the ``NO''-instances are such instances for which the result of the procedure is that all $\lambda_{ij}$ are $0$.

\subsection{A symmetric linear model}

System \ref{eEq4} can be relaxed.
\begin{theorem}
System \ref{eEq4} has solutions iff the following system has solutions:
\begin{equation}
\label{eEq211}
\left \{ \begin{array}{l}
\sum_{i,j} \lambda_{ij}X_{i}O_{G}Z_{j}P_{lk}^{T} = O_{S} \\
\\
\sum_{i,j} \lambda_{ij}X_{i}I_{G}Z_{j}P_{lk}^{T} = I_{S} \\
\\
\sum_{ij}\lambda_{ij} = 1, ~ \lambda_{ij} \geq 0 \\
\end{array} \right.
\end{equation}
\end{theorem}
\begin{proof}
Any solution of system \ref{eEq4} is a solution of system \ref{eEq211}. For any solution of system \ref{eEq211}, there is number $\lambda_{i_1j_1} > 0$. Replacing the number with $1$ and the rest of numbers $\lambda_{ij}$ with $0$ will produce a solution of system \ref{eEq4}. 
\end{proof}
The following symmetric linear program models NP-problems:
$$
\sum_{ij} \lambda_{ij} ~ \rightarrow ~ \min_{\lambda_{ij}},
$$
- under constrains \ref{eEq211}. In full accordance with the Yannakakis theorem \cite{yan}, the system has size $n!k!$. The ellipsoid/separation algorithm \cite{yan, gls} can be deployed to solve the program.
\newline\indent
A necessary condition for compatibility of system \ref{eEq211} is compatibility of the following $(n!+k!)$-size system:
$$
\left \{ \begin{array}{l}
\sum_{j} \lambda_{j}O_{G}Z_{j}P_{lk}^{T} = \sum_i \mu_i X_iO_{S} \\
\\
\sum_{j} \lambda_{j}I_{G}Z_{j}P_{lk}^{T} = \sum_i \mu_i X_iI_{S} \\
\\
\sum_{j}\lambda_{j} = 1, ~ \lambda_{j} \geq 0 \\
\\
\sum_{i}\mu_{i} = 1, ~ \mu_{i} \geq 0 \\
\end{array} \right.
$$

\subsection{A convex program}

\begin{theorem}
Instance $(G,S)$ has solution ``YES'' iff value $\sqrt{n + k}$ is solution of the following program:
\begin{equation}
\label{eEq221}
\sqrt{\sum_{ij} x^2_{ij} + \sum_{ij} z^2_{ij}} ~ \rightarrow ~ \max_{X,Z},
\end{equation}
- under constrains
\begin{equation}
\label{eEq222}
\left \{ \begin{array}{l}
O_{G}ZP_{l k}^{T} = X^T O_{S} \\
\\
I_{G}ZP_{lk}^{T} = X^T I_{S} \\
\\
\sum_{i} x_{ij} = 1, ~ \sum_{j} x_{ij} = 1, ~ x_{ij} \geq 0 \\
\\
\sum_{i} z_{ij} = 1, ~ \sum_{j} z_{ij} = 1, ~ z_{ij} \geq 0 \\
\end{array} \right.
\end{equation}
\end{theorem}
\begin{proof}
The Birkhoff-von Neumann theorem about the double stochastic matrices implies that value $\sqrt{n+k}$ is solution of program \ref{eEq221} under the following constrains:
\begin{equation}
\label{eEq223}
\left \{ \begin{array}{l}
\sum_{i} x_{ij} = 1, ~ \sum_{j} x_{ij} = 1, ~ x_{ij} \geq 0 \\
\\
\sum_{i} z_{ij} = 1, ~ \sum_{j} z_{ij} = 1, ~ z_{ij} \geq 0 \\
\end{array} \right.
\end{equation}
The maximum is reachable then and only then when matrices $X$ and $Z$ are permutation matrices. 
\end{proof}
Convex program \ref{eEq221}/\ref{eEq222} can be solved in polynomial time by the ellipsoid method or by the inner point method \cite{nem1,nem2}.
\newline\indent
Linear polynomial size constrains \ref{eEq222} are asymmetric. Arc/vertex-relabeling of digraph $g$ will rotate vertices of the polytope appropriate to these constrains all over the vertices of the polytope defined by equations \ref{eEq223}.

\subsection{Asymmetric linear models}
\label{s:asymmetric2}

System \ref{eEq222} alone is insufficient for making decisions. It misses the following quadratic constrain:
\begin{equation}
\label{e:cond}
ZP_{lk}^T P_{lk} Z^T \leq U_k.
\end{equation}
\begin{proof}
If permutation matrix $Z$ is the ``arc'' part of a solution of system \ref{eEq222}, then condition \ref{e:cond} is true. On the other hand, let matrices $Z$ and $X$ be such a solution of system \ref{eEq222} that inequality \ref{e:cond} holds. Let's present double stochastic matrix $X$ as a convex combination of permutation matrices:
\[
X = \lambda_1X_1 + \lambda_2X_2 + \ldots, ~ \lambda_1 + \lambda_2 + \ldots = 1,
\]
- where $X_i$ are permutation matrices. Then, system \ref{eEq222} implies:
\[
G \geq O_G ZP_{lk}^T P_{lk} Z^T I_G^T = X^T O_S I_S^T X = X^T S X \geq X_1^T S X_1.
\]
Thus, permutation matrix $X_1^T$ is a solution of inequality \ref{eIn1}.
\end{proof}
\indent
An adequate asymmetric linear model can be build iteratively, as it was done in section \ref{s:asymmetric}. But, let's explore another approach.
\newline\indent
Let's return to symmetric system \ref{eEq211}. Let $\sigma$ be a set of all the matrices participating on the left side of the system:
\[
\sigma = \{ 
\left ( \begin{array}{c}
X_{i}O_{G}Z_{j}P_{lk}^{T} \\
X_{i}I_{G}Z_{j}P_{lk}^{T} \\
\end{array} \right )
~|~ i = 1,2,\ldots,n!; ~ j = 1,2,\ldots,k!\},
\]
- where $X_i$ and $Z_j$ are permutation matrices of size $n\times n$ and $k\times k$, appropriately. Let matrix $C$ be in convex hull of the set:
\[
C \in \mbox{conv}(\sigma)
\]
For certainty, let matrix $C$ be the center of the polytope:
\[
C = \frac{1}{n!k!}\sum_{i,j} \left ( \begin{array}{c}
X_{i}O_{G}Z_{j}P_{lk}^{T} \\
X_{i}I_{G}Z_{j}P_{lk}^{T} \\
\end{array} \right )
= 
\frac{1}{n!k!} \left ( \begin{array}{c}
\sum_i X_{i}O_{G}\sum_j Z_{j}P_{lk}^{T} \\
\sum_i X_{i}I_{G}\sum_j Z_{j}P_{lk}^{T} \\
\end{array} \right ) = 
\]
\[
= \frac{1}{n!k!} \left ( \begin{array}{c}
((n-1)!)_{n\times n}O_{G}((k-1)!)_{k\times l} \\
((n-1)!)_{n\times n}I_{G}((k-1)!)_{k\times l} \\
\end{array} \right )  
= \frac{(1)_{2n\times l}}{n},
\]
- because matrices $O_G$ and $I_G$ have one and only one $1$ per column. Let $\tau$ be the following set:
\[
\tau = \{B ~ | ~ B + C \in \sigma \}.
\]
Polytope $\tau$ is such a shift of polytope $\sigma$ that matrix $C$ gets in the origin of coordinates. Let matrices $B_i$ be a maximal linear independent subsystem of set $\tau$:
\begin{equation}
\label{e:basis}
\begin{array}{l}
L(B_1, B_2, \ldots, B_\beta) = L(\tau) \\
B_i \in \tau, ~ 1 \leq i \leq \beta \\
\end{array},
\end{equation}
- where $L(*)$ is the linear hull of its arguments. Let us emphasize that
\begin{equation}
\label{e:boundary}
\beta \leq 2nl.
\end{equation}
Matrices $B_i + C, ~ i = 1,2,\ldots, \beta$ are a basis of the minimal hyperplane containing set $\sigma$.
\begin{theorem}
\label{t:asymmetric}
NP-instance $(G,S)$ has solution ``YES'' iff the following linear system is compatible:
\begin{equation}
\label{e:asymmetric}
\sum_{i=1}^\beta y_i B_i = \left ( \begin{array}{c}
O_S \\
I_S \\
\end{array} \right ) - C,
\end{equation}
- where numbers $y_i$ are the unknown.
\end{theorem}
\begin{proof}
Necessity. Let instance $(G,S)$ have solution ``YES''. Then, system \ref{eEq211} is compatible. Let's subtract matrix $C$ from both sides of the system; on the left side, let's decompose $C$ over numbers $\lambda_{ij}$; and let's replace the matrices resulting on the left side with their decompositions over basis \ref{e:basis}. The coefficients resulting on the left side are a solution of system \ref{e:asymmetric}.
\newline\indent
Sufficiency. Let's notice that by the definition of set $\tau$,
\[
\forall~\Phi \in L(\tau) ~\exists \mu_{\max} > 0:~\mu \in [0,\mu_{\max}]~ \Rightarrow ~ \mu \Phi \in \mbox{conv}(\tau).
\]
Thus, if system \ref{e:asymmetric} is compatible, then there are such numbers $\mu$ and $\lambda_{ij}$ that the following system is compatible:
\begin{equation}
\label{e:interim}
\left \{ \begin{array}{l}
\sum_{i,j} \lambda_{ij}X_{i}O_{G}Z_{j}P_{lk}^{T} = \mu O_{S} + (\frac{1-\mu}{n})_{n\times l} \\
\\
\sum_{i,j} \lambda_{ij}X_{i}I_{G}Z_{j}P_{lk}^{T} = \mu I_{S} + (\frac{1-\mu}{n})_{n\times l} \\
\\
\sum_{ij}\lambda_{ij} = 1, ~ \lambda_{ij} \geq 0 \\
\end{array} \right.
\end{equation}
By definition, matrices $O_S$ and $I_S$ are $(0,1)$-matrices with one and only one 1 per column. Then, the matrix on the right side of system \ref{e:interim} has only two different elements:
\[
\frac{1-\mu}{n}, ~ \mu + \frac{1-\mu}{n}.
\]
By definition, matrices $O_G$ and $I_G$ are $(0,1)$-matrices with one and only one 1 per column. Then, all $\lambda_{ij}$ on the left side of system \ref{e:interim} will total to two numbers $\xi$ and $\eta$, as well:
\[
\left \{ \begin{array}{l}
\xi = (1-\mu)/n \\
\eta = \mu + (1-\mu)/n \\
(n-1) \xi + \eta = 1 \\
\xi, \eta \geq 0, ~ \mu >0 \\
\end{array} \right.
\]
Among solutions of the last system, there is the boundary solution (when $\mu = \mu_{\max}$):
\[
\mu =1, ~ \eta = 1, ~ \xi = 0.
\]
Substitution of $\mu = 1$ in system \ref{e:interim} produces system \ref{eEq211}. Thus, compatibility of system \ref{e:asymmetric} implies compatibility of system \ref{eEq211}, i.e. solution ``YES'' for instance $(G,S)$.
\end{proof}
Theorem \ref{e:asymmetric} may be seen as a comparison of given digraph $g$ with the complete graph of size $n$. Due to estimation \ref{e:boundary}, system \ref{e:asymmetric} is an asymmetric polynomial size linear model of NP-problems. The asymmetry is due to the selection of basis \ref{e:basis}.
\newline\indent
The selection of basis \ref{e:basis} is a P-problem. The basis selection is reducible to the selection of maximal linear independent subsystems from permutation matrices of sizes $n\times n$ and $k\times k$, appropriately. The matrices are vertices of the appropriate Birkhoff polytopes. Any $(n-1)^2$ different permutation matrices of size $n\times n$ and any $(k-1)^2$ different permutation matrices of size $k\times k$ will produce a polynomial size system for the basis selection. The basis selection is reducible to the solution of a polynomial size system of linear equations. 

\subsection{Hamiltonian graph}

For the Hamiltonian cycle problem \cite{karp,sf,tsp}, the pattern matrix $S$ may be any circular permutation matrix, for example matrix \ref{e:hc}. Let's label arcs of cycle $s$ with the indexes of their end-vertices. Then,
\[
O_S = S, ~ I_S = U_n.
\]
Substitution of these matrices in system \ref{e:asymmetric} gives a polynomial size linear system. Due to theorem \ref{t:asymmetric}, the system is compatible iff digraph $g$ is a Hamiltonian digraph.

\end{document}